\documentclass[runningheads]{llncs}
\usepackage[T1]{fontenc}
\usepackage{graphicx}
\usepackage{hyperref}
\usepackage{color}

\usepackage{comment}

\usepackage{amsmath,amssymb} 
\usepackage{stmaryrd}  

\usepackage{tikz}
\usepackage{pgf} 
\usetikzlibrary{arrows.meta,patterns,automata,arrows,shapes,snakes,topaths,trees,backgrounds,positioning,through,calc}

\newcommand{\comp}{\vartriangleleft}
\newcommand{\compflip}{\vartriangleright}

\usepackage{xcolor}
\definecolor{darkgreen}{rgb}{0.0, 0.7, 0.0}

\newtheorem{fact}{Fact}

\begin{document}
\title{Vagueness and the Connectives}
\author{Wesley H. Holliday\orcidID{0000-0001-6054-9052}}
\authorrunning{Wesley H. Holliday}
\institute{University of California, Berkeley}
\maketitle              % typeset the header of the contribution
\begin{abstract}
Challenges to classical logic have emerged from several sources. According to recent work \cite{Holliday-Mandelkern2022}, the behavior of \textit{epistemic modals} in natural language motivates weakening classical logic to orthologic, a logic originally discovered by Birkhoff and von Neumann \cite{Birkhoff1936} in the study of quantum mechanics. In this paper, we consider a different tradition of thinking that the behavior of \textit{vague predicates} in natural language motivates weakening classical logic to intuitionistic logic or even giving up some intuitionistic principles. We focus in particular on Fine's recent approach to vagueness \cite{Fine2020}. Our main question is: what is a natural non-classical base logic to which to retreat in light of both the non-classicality emerging from epistemic modals and the non-classicality emerging from vagueness? We first consider whether orthologic itself might be the answer. We then discuss whether accommodating the non-classicality emerging from epistemic modals and vagueness might point in the direction of a weaker system of fundamental logic \cite{Holliday2023}.

\end{abstract}
\section{Introduction}\label{Intro}

At the heart of one formulation of the Sorites Paradox is the classical inconsistency of
\begin{equation}
\underset{\mbox{\textsf{extremes}}}{\underbrace{p_1\wedge\neg p_n}}\wedge \underset{\mbox{\textsf{no sharp cutoff}}}{\underbrace{\neg (p_1\wedge\neg p_2)\wedge  \neg (p_2\wedge\neg p_3)\wedge\dots\wedge \neg (p_{n-1}\wedge\neg p_n)}}.\label{SoritesForm}
\end{equation}
For example, for a vague predicate like `young', we would like to say that a person who has been alive for $1$ second is young and that a person who has been alive for 3,000,000,000 seconds\footnote{Approximately 95 years.} is not young, but there is no exact $k$ such that someone who has been alive for $k$ seconds is young while someone who has been alive for $k+1$ seconds is not young. Yet the combination of these claims as in (\ref{SoritesForm}) is inconsistent in classical logic. 

In fact, (\ref{SoritesForm}) is even inconsistent in \textit{intuitionistic} logic.\footnote{This is easy to see using the standard preorder semantics for intuitionistic logic \cite{Dummett1959,Grzegorczyk1964,Kripke1965}: $\neg (p_{k}\wedge\neg p_{k+1})$ is forced at a state $x$ iff for every successor $x'$ of $x$ that forces $p_k$, there is a successor $x''$ of $x'$ that forces $p_{k+1}$, which is a successor of $x$ by transitivity. Thus, if $x$ forces $p_1$, then applying the previous reasoning to each conjunct of \textsf{no sharp cutoff}, it follows that there is a successor of $x$ that forces $p_n$, so $\neg p_n$ cannot also be forced at~$x$.} Those who have thought that the Sorites Paradox motivates moving from classical to intuitionistic logic (e.g., \cite{Wright2001,Bobzien2020}) are apparently satisfied that while \[\textsf{extremes}\vdash_\mathsf{Int}\neg \textsf{no sharp cutoffs},\] at least $\neg \textsf{no sharp cutoff}$ does not intuitionistically entail that there is some sharp cutoff:
\[\neg \textsf{no sharp cutoff}\nvdash_\mathsf{Int}  \underset{\mbox{\textsf{some sharp cutoff}}}{\underbrace{(p_1\wedge\neg p_2)\vee\dots\vee (p_{n-1}\wedge\neg p_n)}}.\]
Yet this subtle position does not satisfy those who would like to make the assertions after (\ref{SoritesForm}) above.

By contrast, Fine \cite{Fine2020} develops an approach to the Sorites that renders (\ref{SoritesForm}) \textit{consistent}. In a fragment of his language with just $\wedge$, $\vee$, and $\neg$, Fine's logic is the logic of bounded distributive lattices equipped with a so-called weak pseudocomplementation (Proposition \ref{FineRep} below).\footnote{A unary operation $\neg$ on a bounded lattice $L$ is a \textit{weak pseudocomplementation} if for all elements $a,b$ of $L$, $a\wedge \neg a=0$, $a\leq \neg\neg a$, and $a\leq b$ implies $\neg b\leq \neg a$. Equivalently, for all elements $a,b$ of $L$, $a\wedge\neg a=0$, and $a\leq \neg b$ implies $b\leq \neg a$.} Over such algebras, (\ref{SoritesForm}) is satisfiable. However, requiring these algebras to also satisfy the law of excluded middle, $a\vee\neg a=1$, turns them into Boolean algebras and hence yields classical logic, rendering (\ref{SoritesForm}) inconsistent. Thus, the rejection of the law of excluded middle is essential to Fine's approach to the Sorites, as it was to Field's \cite{Field2003}.

While Fine's logic of vagueness includes the distributivity of $\wedge$ and $\vee$, Holliday and Mandelkern \cite{Holliday-Mandelkern2022} have recently argued that distributivity is not generally valid, using examples involving the epistemic modals `might' and `must'. They specifically argue for the adoption of \textit{orthologic}, a logic originally discovered by Birkhoff and von Neumann \cite{Birkhoff1936} in the study of quantum mechanics, as the underlying logic for $\wedge$, $\vee$, and $\neg$ in a language with epistemic modals. However, they do not consider whether vagueness might motivate weakening orthologic still further, e.g., by dropping the law excluded middle---which is included in orthologic---when admitting vague predicates into the language. This is the question we wish to address in this paper: what is a natural non-classical base logic to which to retreat in light of both the non-classicality emerging from epistemic modals and the non-classicality emerging from vagueness?

The first candidate we will consider is simply orthologic itself. As we will see, like Fine's logic, orthologic renders (\ref{SoritesForm}) consistent and does so with quite intelligible semantic models,  but it does so thanks to the denial of distributivity rather than excluded middle; though $(p_1\vee\neg p_1)\wedge\dots\wedge (p_n\vee\neg p_n)$ is an orthological theorem applicable to the Sorites Paradox, orthologic does not permit the distributive inference from that conjunction of disjunctions to the disjunction of all conjunctions of the form $\pm p_1\wedge \dots \wedge \pm p_n$, where $\pm$ is empty or $\neg$, which with $p_1\wedge\neg p_n$ and distributivity implies \textsf{some sharp cutoff}. This may seem like an attractively parsimonious solution: perhaps weakening classical logic to orthologic is sufficient to handle both epistemic modals and vagueness. However, there are further desiderata for the logic of vagueness beyond making (\ref{SoritesForm}) consistent, and perhaps these further desiderata motivate giving up excluded middle after all.

This brings us to another candidate, the system we call \textit{fundamental logic}~\cite{Holliday2023}. Fundamental logic weakens orthologic precisely by dropping the principle of double negation elimination that gives rise to excluded middle, as the  phenomena involving vagueness might motivate us to do, and weakens Fine's logic precisely by dropping distributivity, as the phenomena involving epistemic modals might motivate us to do. In \cite{Holliday2023}, fundamental logic is defined in terms of a Fitch-style natural deduction system including only introduction and elimination rules for the connectives $\wedge$, $\vee$, and $\neg$. In particular, this natural deduction system does not include the rule of Reductio Ad Absurdum (if the assumption of $\neg\varphi$ leads to a contradiction, then conclude $\varphi$) or the rule that Fitch called Reiteration (which allows pulling certain previously derived formulas into subproofs). Modulo the introduction and elimination rules, Reductio Ad Absurdum is the culprit behind double negation elimination and excluded middle, while Reiteration is the culprit behind distributivity. By including only the introduction and elimination rules for the connectives, fundamental logic avoids those principles and the associated problems involving vague predicates and epistemic modals.

The rest of the paper is organized as follows. In \S~\ref{LogicsSection}, we define the logics that we will consider as candidate logics for a language with $\wedge$, $\vee$, $\neg$, and vague predicates. In \S~\ref{FixpointSemantics}, we briefly review a general relational semantics for orthologic, fundamental logic, intuitionistic logic, and classical logic. In \S~\ref{FineSemantics}, we review a different semantics that Fine uses for the purposes of his logic. We then turn to \textit{models} of the Sorites series: first a model appropriate for orthologic or Fine's logic (depending on the semantics applied to the model) in \S~\ref{SymmSorites} and then a model appropriate for fundamental logic in \S~\ref{PseudosymmSorites}. These models show that orthologic and fundamental logic, which can handle the non-distributivity arising from epistemic modals, can also handle the Sorites series at least in some minimal sense. In \S~\ref{Discussion}, we discuss whether we should favor an orthological approach to the Sorites or a fundamental one. We conclude in \S~\ref{Conclusion} with some open questions.

\section{Logics}\label{LogicsSection}

We will focus on the logical connectives of negation, conjunction, and disjunction, using the following formal language.

\begin{definition}\label{LangDef} \textnormal{Given a nonempty set $\mathsf{Prop}$ of propositional variables, let $\mathcal{L}$ be the propositional language generated by the grammar
\[\varphi::= p\mid \neg\varphi\mid (\varphi\wedge\varphi)\mid (\varphi\vee\varphi)\]
where $p\in\mathsf{Prop}$.}\end{definition}

The logics we will consider for this language are all examples of ``intro-elim logics'' in the following sense from \cite{Holliday2023}.

\begin{definition}\label{LogicDef} \textnormal{An \textit{intro-elim logic} is a binary relation $\vdash\,\subseteq\mathcal{L}\times\mathcal{L}$ such that for all $\varphi,\psi,\chi\in\mathcal{L}$:
\begin{center}
\begin{tabular}{ll}
1. $\varphi\vdash\varphi$ & 8. if $\varphi\vdash\psi$ and $\psi\vdash\chi$, then $\varphi\vdash\chi$  \\
2. $\varphi\wedge\psi\vdash\varphi$ & \\
3.  $\varphi\wedge\psi\vdash\psi$ & 9. if $\varphi\vdash \psi$ and $\varphi\vdash\chi$, then $\varphi\vdash \psi\wedge\chi$\\
4. $\varphi\vdash \varphi\vee\psi$ & \\
5. $\varphi\vdash \psi\vee\varphi$ & 10. if $\varphi\vdash\chi$ and $\psi\vdash\chi$, then $\varphi\vee\psi \vdash \chi$ \\
6. $\varphi\vdash \neg\neg\varphi$ &  \\
7. $\varphi\wedge\neg\varphi\vdash\psi$ \qquad\qquad\qquad& 11. if  $\varphi\vdash\psi$, then $\neg\psi\vdash\neg\varphi$. 
\end{tabular}
\end{center}
As in  \cite{Holliday2023}, we call the smallest intro-elim logic \textit{fundamental logic}, denoted $\vdash_\mathsf{F}$.}
\end{definition}

\textit{Orthologic} \cite{Goldblatt1974} is obtained from fundamental logic by adding
\begin{itemize}
\item \textit{double negation elimination}: $\neg\neg\varphi\vdash\varphi$. 
\end{itemize}
\textit{Compatibility logic} \cite{Fine2020} in the $\{\wedge,\vee,\neg\}$-fragment  is obtained from fundamental logic by strengthening proof-by-cases in Definition \ref{LogicDef}.10 to
\begin{itemize}
\item  \textit{proof-by-cases with side assumptions}: if $\alpha\wedge\varphi\vdash\chi$ and $\alpha\wedge\psi\vdash \chi$, then $\alpha\wedge (\varphi\vee\psi)\vdash\psi$.
\end{itemize}
\textit{Intuitionistic logic} in the $\{\wedge,\vee,\neg\}$-fragment \cite{Rebagliato1993} is obtained from compatibility logic by  
strengthening Definition \ref{LogicDef}.6/11 to the following, where $\bot$ abbreviates a contradiction such as $p\wedge \neg p$:
\begin{itemize}
\item \textit{psuedocomplementation}: if $\varphi\wedge\psi\vdash \bot$, then $\psi\vdash\neg\varphi$,
\end{itemize}
\textit{Classical logic} is obtained by strengthening orthologic with either proof by cases with side assumptions or pseudocomplementation (see \cite[Proposition~3.7]{Holliday-Mandelkern2022}).

In the context of a language with epistemic modals $\Diamond$ (`might') and $\Box$ (`must'), Holliday and Mandelkern \cite{Holliday-Mandelkern2022} argue that proof-by-cases with side assumptions and pseudocomplementation are both invalid. They start by accepting well-known arguments (see Section 2 of  \cite{Holliday-Mandelkern2022} for a review of the relevant data) that sentences of the form `It's raining and it might not be raining' ($p\wedge\Diamond\neg p$) are semantically contradictory, not just pragmatically infelicitous, because they embed under other operators in the way that outright contradictions (`It's raining and it's not raining') do, rather than how consistent but infelicitous Moore sentence (`It's raining but I don't know it') do. Thus, if $\Diamond$ is the epistemic modal `might', we have $p\wedge\Diamond \neg p\vdash\bot$.\footnote{Of course if we consider an alethic or counterfactual possibility modal $M$ such as `it could have been the case that', then $p\wedge M \neg p\nvdash\bot$.} But then pseudocomplementation would yield $\Diamond\neg p\vdash \neg p$, which is absurd, since we cannot validly reason from `It might not be raining' to `It's not raining'. And proof-by-cases with side assumptions would yield $\Diamond p \wedge (p\vee\neg p)\vdash p$ (since $\Diamond p\wedge p\vdash p$ and $\Diamond p\wedge \neg p\vdash\bot\vdash p$), which is similarly absurd, since we cannot validly reason from `It might be raining' together with `It's raining or it's not raining' to `It \textit{is} raining'.\footnote{Note that this argument does not assume that one accepts the premise `It's raining or it's not raining' as a result of a general acceptance of excluded middle. The point is just that accepting this as a meteorological fact on a particular occasion does not license the inference from `It might be raining' to `It is raining'.}

Although we do not have the epistemic modals $\Diamond$ and $\Box$ in our formal language $\mathcal{L}$, we intend the elements of $\mathsf{Prop}$ from which $\mathcal{L}$ is generated to be genuine \textit{propositional variables}, standing in for arbitrary propositions (cf.~\cite[pp.~147-8]{Burgess2003}). Thus, if we accept the counterexamples to pseudocomplementation and proof-by-cases with side assumptions above, then we cannot accept these principles for~$\mathcal{L}$. This suggests orthologic as an appropriate base logic, as argued in \cite{Holliday-Mandelkern2022}. But when we consider reasoning not only with epistemic modals but also with vague predicates, might this suggest going still weaker than orthologic? This is one of the questions we will consider in what follows.

\section{Relational semantics for ortho and fundamental logic}\label{FixpointSemantics}

To give semantics for the logics introduced in \S~\ref{LogicsSection}, we begin with a familiar starting point from modal logic, namely relational frames.

\begin{definition}\label{ModelDef} \textnormal{A \textit{relational frame} is a pair $(X,\comp)$ where $X$ is a nonempty set and $\comp$ is a binary relation on $X$. A \textit{relational model} for $\mathcal{L}$ is a triple $\mathcal{M}=(X,\comp, V)$ where $(X,\comp)$ is a relational frame and $V$ assigns to each $p\in\mathsf{Prop}$ a set $V(p)\subseteq X$.}\end{definition}

In possible world semantics for intuitionistic logic, $\comp$ would be a preorder $\leqslant$, and $V(p)$ would be required to be a  \textit{downset}, i.e., if $x\in V(p)$ and $x'\leqslant x$, then $x'\in V(p)$. States in $X$ are understood as information states, and $x'\leqslant x$ means that $x'$ contains all the information that $x$ does and possibly more.\footnote{Many authors work instead with upsets and take $x'\geqslant x$ to mean that $x'$ contains all the information that $x$ does and possibly more. Our formulation in terms of downsets follows, e.g., Dragalin~\cite{Dragalin1988}.}

We will give a different interpretation of $\comp$ that motivates a different constraint on $V(p)$. We read $x\comp y$ as \textit{$x$ is open to $y$}. Paraphrasing Remark 4.2 of \cite{Holliday2023}, we understand openness as follows:
\begin{itemize}
\item There is a distinction between \textit{accepting} a proposition and \textit{rejecting} it.
\item It should be possible for a partial state to be completely noncommittal about a proposition, so non-acceptance of $A$ should not entail rejection of $A$.
\item It should be possible for a state to reject a proposition without accepting its negation; e.g., an intuitionist might reject an instance of the law of excluded middle but will not accept its negation, which is a contradiction (see Field's \cite{Field2003} separation of rejection, non-acceptance, and acceptance of the negation).
\item We say that $x$ is \textit{open to} $y$ iff $x$ does not reject any proposition that $y$ accepts.
\end{itemize}

This picture motivates the following definition.
\begin{definition}\label{AcceptRejectDef} \textnormal{Given a relational frame $(X,\comp)$ and  $A\subseteq X$, we say that
\begin{enumerate}
\item $x$ \textit{accepts} $A$ if $x\in A$; 
\item\label{AcceptRejectDef2} $x$ \textit{rejects} $A$ if for all $y\compflip x$, $y\not\in A$; 
\item $x$ \textit{accepts the negation of} $ A$ if for all $y\comp x$, $y\not\in A$. 
\end{enumerate}}
\end{definition}

Now we can state the constraint on $V(p)$ that we will use instead of the downset constraint:
\[\mbox{if $x$ does not accept $V(p)$, then there is a state $x'$ open to $x$ that rejects $V(p)$,}\]
or formally,
\begin{equation}\mbox{if }x\not\in V(p),\mbox{ then }\exists x'\comp x\, \forall x''\compflip x' \; x''\not\in V(p).\label{FixEq}\end{equation}

\begin{definition} \textnormal{We say a relational model $\mathcal{M}=(X,\comp, V)$ is a \textit{relational fixpoint model} if for all $p\in\mathsf{Prop}$, the set $V(p)$ satisfies (\ref{FixEq}) for all $x\in X$.}
\end{definition}

The motivation for the `fixpoint' terminology is the following. Given a frame $(X,\comp)$, we define a closure operator\footnote{A function $c:\wp(X)\to\wp(X)$ is a \textit{closure operator} if for all $A,B\subseteq X$, we have $A\subseteq c(A)$, $c(c(A))=c(A)$, and $A\subseteq B$ implies $c(A)\subseteq c(B)$.} ${c_\comp\colon \wp(X)\to\wp(X)}$  by \[c_\comp(A)=\{x\in X\mid \forall x'\comp x\,\exists x''\compflip x'\colon x''\in A \}.\]
A \textit{fixpoint} of $c_\comp$ is a set $A\subseteq X$ such that $c_\comp(A)=A$. Since $A\subseteq c_\comp(A)$ for any $A\subseteq X$, that $A$ is a fixpoint is equivalent to $c_\comp (A)\subseteq A$. The contrapositive of this is just (\ref{FixEq}) for $A$ in place of $V(p)$:  \[\mbox{if }x\not\in A\mbox{, then }\exists x'\comp x\,\forall x''\compflip x'\colon x''\not\in A.\]

\noindent If we define a \textit{proposition} in $(X,\comp)$ to be a fixpoint of $c_\comp$, we have the following.

\begin{lemma} \textnormal{In a relational frame $(X,\comp)$, $x\comp y$ iff $x$ does not reject any proposition that $y$ accepts.}
\end{lemma}
\begin{proof} We repeat the proof from Footnote 14 of \cite{Holliday2023}. If $x\comp y$ and $y$ accepts $A$, so $y\in A$, then $x$ does not reject $A$ by Definition \ref{AcceptRejectDef}.\ref{AcceptRejectDef2}. Conversely, suppose $x\not\comp y$ Now $y$ accepts  $c_\comp(\{y\})$, which is a fixpoint of $c_\comp$. But we claim $x$ rejects $c_\comp(\{y\})$. For if there is some $z$ such that $x\comp z\in c_\comp(\{y\})$, then by definition of $c_\comp$, there is an $x'\compflip x$ such that $x'\in \{y\}$, so $x'=y$. But this contradicts $x\not\comp y$. Hence there is no $z$ such that $x\comp z\in c_\comp(\{y\})$, which means that $x$ rejects $c_\comp(\{y\})$.  \qed
\end{proof}

Finally, we define the semantics of the connectives. The semantics of $\neg$ and $\wedge$ look the same as in relational semantics for intuitionistic logic (though again $\comp$ need not be a preorder), but the semantics of $\vee$ is different. Instead of interpreting $\vee$ as the union of propositions, as in intuitionistic logic, we interpret $\vee$ as the \textit{closure} of the union, i.e., by applying $c_\comp$ to the union.

\begin{definition}\label{TruthDef} \textnormal{Given a relational fixpoint model $\mathcal{M}=(X,\comp, V)$ for $\mathcal{L}$, $x\in X$, and $\varphi\in\mathcal{L}$, we define $\mathcal{M},x\Vdash \varphi$ as follows:
\begin{enumerate}
\item\label{TruthDefAtom} $\mathcal{M},x\Vdash p$ iff $x\in V(p)$;
\item $\mathcal{M},x\Vdash\neg\varphi$ iff  $\forall x'\comp x$, $\mathcal{M},x'\nVdash \varphi$;
\item $\mathcal{M},x\Vdash \varphi\wedge\psi$ iff $\mathcal{M},x\Vdash\varphi$ and $\mathcal{M},x\Vdash\psi$;
\item\label{TruthDefOr} $\mathcal{M},x\Vdash \varphi\vee\psi$ iff $\forall x'\comp x$ $\exists x''\compflip x'$: $\mathcal{M},x''\Vdash \varphi$ or $\mathcal{M},x''\Vdash \psi$.
\end{enumerate}
Let $\llbracket \varphi\rrbracket^\mathcal{M}=\{x\in X\mid \mathcal{M},x\Vdash\varphi\}$.}
\end{definition}

The idea behind the clause for $\vee$ is that in order for $x$ to accept $\varphi\vee\psi$, it is not required that $x$ already accepts one of the disjuncts---instead what is required is that \textit{no state open to $x$ rejects both disjuncts}. (Although it may seem that we hereby take a stand against intuitionistic logic on $\vee$, see Theorem \ref{SpecialCases}.\ref{SpecialCases1}.)

We can prove by an easy induction on the structure of formulas that the semantic value of every formula is a proposition in the following sense.

\begin{lemma}\label{FixpointLemma} \textnormal{Let $\mathcal{M}=(X,\comp,V)$ be a relational fixpoint model. Then for any $\varphi\in\mathcal{L}$, the set $\llbracket \varphi\rrbracket^\mathcal{M}$ is a  
fixpoint of $c_\comp$.}
\end{lemma}

A key feature of the states in a relational fixpoint model is that they can be \textit{partial}---there may be some propositions that they neither accept nor reject. Thus, we can consider how the set of propositions accepted by one state relates to that of another state. For this we have a key definition and lemma.

\begin{definition}\textnormal{Given a relational fixpoint model $\mathcal{M}=(X,\comp, V)$ and $x,y\in X$, we say that $x$ \textit{pre-refines} $y$ if for every $z\comp x$, we have $z\comp y$.}
\end{definition}
\noindent In symmetric models, we simply say \textit{refines} instead of \textit{pre-refines}. An easy induction on the structure of formulas establishes the following.

\begin{lemma} \textnormal{Given a relational fixpoint model $\mathcal{M}=(X,\comp, V)$ and $x,y\in X$, if $x$ pre-refines $y$, then for any $\varphi\in\mathcal{L}$, $\mathcal{M},y\Vdash\varphi$ implies $\mathcal{M},x\Vdash\varphi$.}
\end{lemma}

The final piece of setup we need are the definition of consequence over a class of models, as well as of soundness and completeness of a  logic with respect to the class, all of which is completely standard.

\begin{definition} \textnormal{Given a class $\mathsf{C}$ of relational models and $\varphi,\psi\in\mathcal{L}$, we define $\varphi\vDash_\mathsf{C}\psi$ iff for all $\mathcal{M}\in\mathsf{C}$ and states $x$ in $\mathcal{M}$, if $\mathcal{M},x\Vdash\varphi$, then $\mathcal{M},x\Vdash\psi$.}

\textnormal{We say that an intro-elim logic $\vdash_\mathsf{L}$ is \textit{sound} (resp.~\textit{complete}) \textit{with respect to $\mathsf{C}$} if for all $\varphi,\psi\in\mathcal{L}$, $\varphi\vdash_\mathsf{L}\psi$ (resp.~$\varphi\vDash_\mathsf{C}\psi$) implies $\varphi\vDash_\mathsf{C}\psi$ (resp.~$\varphi\vdash_\mathsf{L}\psi$).}
\end{definition}

We are now prepared to state soundness and completeness theorems for four of the logics introduced in \S~\ref{LogicsSection}. We begin with orthologic and fundamental logic.

\begin{theorem}[\cite{Goldblatt1974}]\label{OrthoSoundComplete} \textnormal{Orthologic is sound and complete with respect to the class of relational fixpoint models in which $\comp$ is reflexive and \textit{symmetric}.}
\end{theorem}

\begin{theorem}[\cite{Holliday2023}]\label{FundamentalSoundComplete} \textnormal{Fundamental logic is sound and complete with respect to the class of relational fixpoint models in which $\comp$ is reflexive and \textit{pseudosymmetric}: if $y\comp x$, then there is a $z\comp y$ that pre-refines $x$.}
\end{theorem}

The key property of pseudosymmetry of $\comp$ for fundamental logic is equivalent to the intuitive condition that \textit{if $x$ accepts a proposition $\neg A$, then $x$ rejects $A$} (see \cite[\S~4.1]{Holliday2023}). Strengthening pseudosymmetry to symmetry implies the non-constructive condition that if $x$ accepts $\neg\neg A$, then $x$ accepts $A$.

As for intuitionistic and classical logic, it is easy to see that standard semantics for these logics are simply special cases of the relational fixpoint semantics (see \cite[\S~2]{Holliday2022} for further explanation), so from well-known completeness theorems for the logics with respect to their standard semantics, we obtain the following.

\begin{theorem}\label{SpecialCases}$\,$\textnormal{
\begin{enumerate}
\item\label{SpecialCases1} (Possible world semantics for intuitionistic logic) Intuitionistic logic is sound and complete with respect to the class of relational fixpoint models in which $\comp$ is reflexive and \textit{transitive}.
\item\label{SpecialCases2} (Possibility semantics for classical logic) Classical logic is sound and complete with respect to the class of relational fixpoint models in which $\comp$ is reflexive, symmetric, and \textit{compossible} \cite{Holliday2022}:  whenever $x\comp y$, there is some $z$ that refines both $x$ and $y$.
\item (Possible world semantics for classical logic) Classical logic is sound and complete with respect to the class of relational fixpoint models in which $\comp$ is the identity relation (and $|X|=1$).
\end{enumerate}}
\end{theorem}

\section{Relational semantics for compatibility logic}\label{FineSemantics}

To give semantics for his \textit{compatibility logic}, Fine \cite{Fine2020} does not use the relational \textit{fixpoint} models of \S~\ref{FixpointSemantics}.\footnote{It is possible to give a semantics for compatibility logic using fixpoint models, using results from \cite[\S~3.1]{Massas2023} and \cite{Holliday2021,Holliday2023}, but we will not go into the details here.} For he imposes no constraints on the valuation function $V$ (and hence no constraints on what sets of states count as propositions) and uses a different interpretation of $\vee$, as follows.

\begin{definition}\label{FineTruth} \textnormal{Given a relational model $\mathcal{M}=(X,\comp, V)$ for $\mathcal{L}$, $x\in X$, and $\varphi\in\mathcal{L}$, we define $\mathcal{M},x\vDash \varphi$ as follows:
\begin{enumerate}
\item $\mathcal{M},x\vDash p$ iff $x\in V(p)$;
\item $\mathcal{M},x\vDash\neg\varphi$ iff  $\forall x'\comp x$, $\mathcal{M},x'\nvDash \varphi$;
\item $\mathcal{M},x\vDash \varphi\wedge\psi$ iff $\mathcal{M},x\vDash\varphi$ and $\mathcal{M},x\vDash\psi$;
\item $\mathcal{M},x\vDash \varphi\vee\psi$ iff $\mathcal{M},x\vDash\varphi$ or $\mathcal{M},x\vDash\psi$.
\end{enumerate}
Let $\|\varphi\|^\mathcal{M}=\{x\in X\mid \mathcal{M},x\vDash\varphi\}$.}
\end{definition}

The soundness of compatibility logic with respect to this semantics is easy to check. Note, in particular, from the clauses for $\wedge$ and $\vee$ that this semantics validates the distributive laws, or equivalently, proof-by-cases with side assumptions, which are invalid according to orthologic and fundamental logic. However, if we ignore disjunction, then of course Fine's semantic clauses are the same as those in \S\ref{FixpointSemantics}. An obvious induction yields the following.

\begin{lemma}\label{EquivSem} \textnormal{For any relational fixpoint model $\mathcal{M}=(X,\comp, V)$  and $\varphi\in\mathcal{L}$, if $\varphi$ does not contain $\vee$, then $\llbracket \varphi\rrbracket^\mathcal{M}=\|\varphi\|^\mathcal{M}$.}\end{lemma}

Now Fine does not work with \textit{arbitrary} relational models. He assumes the relation $\comp$ is reflexive and symmetric. Under these assumptions, we will write `$\between$' instead of `$\comp$' to emphasize the symmetry. Reflexivity and symmetry are the same conditions on the relation  as for orthologic in Theorem \ref{OrthoSoundComplete}, but the key difference is again that Fine puts no constraints on the valuation function $V$, so his semantics does not validate orthologic even in the $\vee$-free fragment. In particular, it does not validate the inference from $\neg\neg\varphi$ to $\varphi$, which othologic~does.

We will sketch the proof that compatibility logic is complete with respect to the class of relational models in which the relation is reflexive and symmetric. Those familiar with algebraic logic will have no trouble filling in all the details. To state the key result, let us say that a unary operation $\neg$ on a bounded lattice $L$ is a \textit{weak pseudocomplementation} (terminology from \cite{Dzik2006,Dzik2006b,Almeida2009}) if for all $a,b\in L$, we have that $a\wedge\neg a=0$, that $a\leq \neg\neg a$, and that $a\leq b$ implies $\neg b\leq\neg a$.

\begin{theorem}\label{FineRep} \textnormal{For any bounded distributive lattice $L$ equipped with a weak pseudocomplementation $\neg$, there is a set $X$ and a reflexive and symmetric binary relation $\between$ on $X$ such that $(L,\neg)$ embeds into $(\wp(X),\neg_{\between})$,  where $\neg_{\between} A=\{x\in X\mid \mbox{for all } y\between x,\, y\not\in A\}$.}
\end{theorem}
\begin{proof}[Sketch] Apply Stone's \cite{Stone1938} representation of distributive lattices using prime filters to  $L$. Then given prime filters $F,F'$, say that $F\between F'$ if there is no $a\in L$ such that $a\in F$ and $\neg a\in F'$.
\qed\end{proof}

\begin{theorem}\label{CompatibilityCompleteness} \textnormal{Compatibility logic is sound and complete with respect to reflexive and symmetric models under the semantics of Definition \ref{FineTruth}.}
\end{theorem}
\begin{proof}[Sketch] As noted above, soundness is an easy check. For completeness, apply Theorem \ref{FineRep} to the Lindenbaum-Tarski algebra of compatibility~logic.
\qed\end{proof}

Weak pseudocomplementations are precisely the types of negations in algebras for fundamental logic (see \cite[\S~3]{Holliday2023}), only the underlying lattices of these algebras are not necessarily distributive. However, if we drop $\vee$ from the language, then we cannot notice this difference between compatibility logic and fundamental logic. The following theorem is essentially Proposition 4.34 of \cite{Holliday2023} but rephrased in terms of compatibility logic instead of the modal logic \textbf{KTB}.

\begin{proposition} \textnormal{For any  $\varphi,\psi\in\mathcal{L}$ not containing $\vee$, we have $\varphi\vdash\psi$ in compatibility logic if and only if $\varphi\vdash\psi$ in fundamental logic.}
\end{proposition}
\noindent However, the analogous result does not hold for formulas not containing $\wedge$: as Guillaume Massas (personal communication) pointed out, \[\neg ( \neg \varphi \vee \neg (\psi \vee \chi))\vdash\neg (\neg \varphi \vee \neg \psi) \vee \neg (\neg \varphi \vee \neg \chi)\] in compatibility logic but not in fundamental logic, as a consequence of distributivity in the former but not the latter.

Having completed our brief tour of the logics and semantics of interest, we turn in the next section back to vagueness.

\section{The symmetric Sorites model}\label{SymmSorites}

Recall the Sorites Paradox from \S~\ref{Intro}. In this section, we define a relational fixpoint model for each Sorites series. The construction is inspired by a passage of Fine \cite[p.~42]{Fine2020}:
\begin{quote} [T]hink of the points of a model as corresponding to different admissible uses of the language. The reflexive and symmetric relation will then relate two admissible uses when they are \textit{compatible} in the sense that there is no conflict in what is true under the one use and what is true under the other. Suppose, for example, that there are 100 men in our sorites series. Then the use in which the first 30 men are taken to be bald and the last 50 men are taken to be not bald will be compatible with the use in which the first 31 men are taken to be bald and the last 49 are taken not to be bald.
\end{quote}

First, let us specialize our propositional language for the Sorites paradox in particular. For convenience, we now take the first item in the Sorites series to be labeled by $0$ instead of $1$, and we adopt the identification $n=\{0,\dots,n-1\}$. 

\begin{definition} \textnormal{For $n\in\mathbb{N}$, let $\mathcal{L}_n$ be defined like $\mathcal{L}$ in Definition \ref{LangDef} with $\mathsf{Prop}=\{p_k\,\mid \, k\,\in\, n\}$.}
\end{definition}

We take $p_k$ to mean that the $k$-th member of the Sorites series satisfies the relevant vague predicate (e.g., `young', `bald', etc.).

\begin{definition}\label{Sorites} \textnormal{For $n,\delta\in\mathbb{N}$ with $1\leq\delta < n-1$, we define the \textit{symmetric Sorites model} $\mathcal{S}_{n,\delta}=(S,\between,V)$ for $\mathcal{L}_n$ as follows:
\begin{enumerate}
\item\label{SoritesA} $S$ is the set of all pairs $( i,j )$ for $i,j\in n\cup \{ -\infty,\infty\}$ such that $i+\delta<j$;\footnote{Note that $-\infty +\delta =-\infty$.}
\item\label{SoritesB} $( i,j )\between ( i',j' )$ iff $\mathrm{max}(i,i')< \mathrm{min} (j,j')$;
\item\label{SoritesC} $V(p_k) = \{ ( i,j ) \in S\mid k\leq i\}$.
\end{enumerate}}
\end{definition}
\noindent Figure \ref{SymModel} shows the symmetric Sorites model for $n=4$ and $\delta=1$.

The intuition behind Definition \ref{Sorites} is as follows: state $(i,j)$ will make $p_k$ true for all $k\leq i$, by part \ref{SoritesC} of the definition,  and will make $\neg p_k$ true for all $k\geq j$ by Fact \ref{NegAtom} below. Thus, part \ref{SoritesB} says that two states are compatible when there is no $p_k$ such that one state makes $p_k$ true and the other makes $\neg p_k$ true. Finally, note that part \ref{SoritesA} rules out sharp cutoffs: not only can a state not make $p_k$ true and $\neg p_{k+1}$ true, but if we pick a $\delta$ greater than $1$, then we can even rule out a state making $p_k$ true and $\neg p_{k+2}$ true, etc. For example, we may not want to admit a state according to which a person who is $k$ seconds old is young but a person who is $k+2$ seconds old is not young. No doubt the appropriate choice of $\delta$ in a given context is itself vague. Also note the role of $-\infty$ and $\infty$ in part~\ref{SoritesA}: the state $(-\infty,\infty)$ will not make true any formulas of the form $p_k$ or $\neg p_k$; a state such as $(-\infty, j)$ will not make true any formulas of the form $p_k$; and a state such as $(j,\infty)$ will not make true any formulas of the form $\neg p_k$.

\begin{figure}[h]
\begin{center}
\begin{tikzpicture}[scale=.75, every node/.style={scale=.75}][->,>=stealth',shorten >=1pt,auto,node distance=2.8cm,
                    semithick]

  \tikzstyle{every state}=[fill=white,draw=black,text=black]

  \node[state] (A) at (-6,0)                    [label=below:{$(3,\infty)$}] {{}};

  \node[state] (B) at (6,0)  [label=below:{$(-\infty, 0)$}] {{}};

  \node[state] (C) at (0,0) [label=below:{$(-\infty, \infty)$}] {{}};

  \node[state] (D) at (6,4)       [label=right:{$(-\infty, 1)$}] {{}};

  \node[state] (E) at (3,5)       [label=right:{$(-\infty, 2)$}] {{}};

  \node[state] (F) at (2,2)  [label=below:{$(-\infty, 3)$}] {{}};

  \node[state] (G) at (-2,2)  [label=below:{$(0, \infty)$}] {{}};

  \node[state] (H) at (2, 8)       [label=right:{$(0, 2)$}] {{}};

  \node[state] (I) at (0,10)  [label=above:{$(0, 3)$}] {{}};

  \node[state] (J) at (-3, 5)      [label=left:{$(1, \infty)$}] {{}};

  \node[state] (K) at (-2, 8)     [label=left:{$(1, 3)$}] {{}};

  \node[state] (L) at (-6, 4)      [label=left:{$(2, \infty)$}] {{}};

\path 
    (A) edge[-]              node {} (C)
        edge[-]              node {} (G)
        edge[-]              node {} (J)
        edge[-]              node {} (L)
    (B) edge[-]              node {} (C)
        edge[-]              node {} (D)
        edge[-]              node {} (E)
        edge[-]              node {} (F)
    (C) edge[-, bend right]  node {} (D)
        edge[-]              node {} (E)
        edge[-]              node {} (F)
        edge[-]              node {} (G)
        edge[-]              node {} (H)
        edge[-]              node {} (I)
        edge[-]              node {} (J)
        edge[-]              node {} (K)
        edge[-, bend left]   node {} (L)
    (D) edge[-]              node {} (E)
        edge[-]              node {} (F)
        edge[-]              node {} (G)
        edge[-]              node {} (H)
        edge[-, bend right=45] node {} (I)
    (E) edge[-]              node {} (F)
        edge[-]              node {} (G)
        edge[-]              node {} (H)
        edge[-]              node {} (I)
        edge[-]              node {} (J)
        edge[-]              node {} (K)
    (F) edge[-]              node {} (G)
        edge[-]              node {} (H)
        edge[-]              node {} (I)
        edge[-]              node {} (J)
        edge[-]              node {} (K)
        edge[-]              node {} (L)
    (G) edge[-]              node {} (H)
        edge[-]              node {} (I)
        edge[-]              node {} (J)
        edge[-]              node {} (K)
        edge[-]              node {} (L)
    (H) edge[-]              node {} (I)
        edge[-]              node {} (J)
        edge[-]              node {} (K)
    (I) edge[-, bend right=45] node {} (L)
        edge[-]              node {} (J)
        edge[-]              node {} (K)
    (J) edge[-]              node {} (K)
        edge[-]              node {} (L)
    (K) edge[-]              node {} (L);

\end{tikzpicture}
\end{center}
\caption{The symmetric Sorites model for $n=4$ and $\delta=1$.}\label{SymModel}
\end{figure}
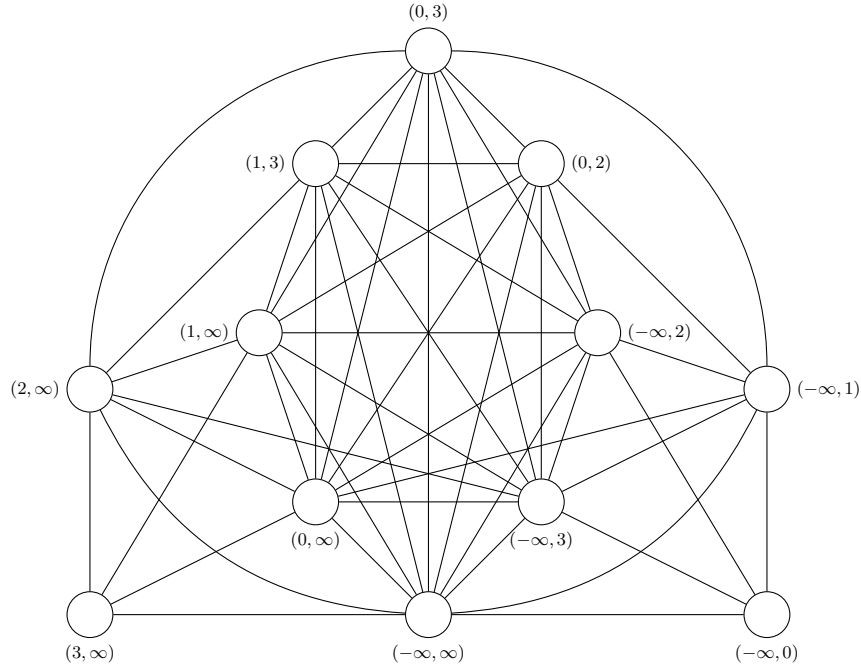

A key point about the symmetric Sorites model is that it is a model not only for Fine's semantics from \S~\ref{FineSemantics} but also for the fixpoint semantics from \S~\ref{FixpointSemantics}.

\begin{proposition}\label{Fixpoints} \textnormal{$\mathcal{S}_{n,\delta}$ is a relational fixpoint model.}
\end{proposition}
\begin{proof} We must show that $V(p_k)$ is a fixpoint of $c_\between$. Given $(i,j)\in S$, we have $i<j$ by Definition~\ref{Sorites}.\ref{SoritesA}. Now suppose $(i,j)\not\in V(p_k)$, so $i<k$ by Definition~\ref{Sorites}.\ref{SoritesC}. Let $(i',j')=(-\infty, k)$, so $(i',j')\in S$ by Definition \ref{Sorites}.\ref{SoritesA}. Then $\mathrm{max}(i,i')=i<\mathrm{min}(j,j')$, so $(i',j')\between (i,j)$ by Definition \ref{Sorites}.\ref{SoritesB}. Now for any $(i'',j'')\between (i',j')$, we have $i''<j'=k$ by Definition \ref{Sorites}.\ref{SoritesB}, so $(i'',j'')\not\in V(p_k)$ by Definition \ref{Sorites}.\ref{SoritesC}. Thus, we have shown there is  an $(i',j')\between (i,j)$ such that for all $(i'',j'')\between (i',j')$, we have $(i'',j'')\not\in V(p_k)$, which shows that $V(p_k)$ is a fixpoint of $c_\between$.\qed\end{proof}

The upshot of Proposition \ref{Fixpoints} is that $\mathcal{S}_{n,\delta}$ is a model for orthologic.

\begin{corollary}\label{SoundCor}\textnormal{Since the relation $\between$ in the relational fixpoint model $\mathcal{S}_{n,\delta}={(S,\between,V)}$ is reflexive and symmetric, orthologic is sound with respect to $\mathcal{S}_{n,\delta}$ according to the semantics of Definition \ref{TruthDef}, while compatibility logic is sound with respect to $\mathcal{S}_{n,\delta}$ according to the semantics of Definition~\ref{FineTruth}.}
\end{corollary}
\begin{proof} That $\between$ is reflexive and symmetric is immediate from Definition \ref{Sorites}.\ref{SoritesB}. Then apply  Theorems~\ref{OrthoSoundComplete} and \ref{CompatibilityCompleteness}.
\qed\end{proof}

We now build up a sequence of facts about $\mathcal{S}_{n,\delta}$. The first is immediate from  Definitions \ref{TruthDef}.\ref{TruthDefAtom} and  \ref{Sorites}.\ref{SoritesC}.

\begin{fact}\label{Atom} \textnormal{For any $k\in n$, $\llbracket  p_k\rrbracket^{\mathcal{S}_{n,\delta}}  = \{ ( i,j ) \in S\mid k\leq i\}$.} 
\end{fact}

We can prove an analogous fact for $\neg p_k$ and the \textit{second} coordinate of states in the model.

\begin{fact}\label{NegAtom} \textnormal{For any $k\in n$, $\llbracket \neg p_k\rrbracket^{\mathcal{S}_{n,\delta}} = \{(i,j)\in S\mid  j\leq k\}$.}
\end{fact}
\begin{proof} Given $(i,j)\in S$, suppose $j\leq k$. Then for any $(i',j')\between (i,j)$, we have  $i'<j\leq k$ by Definition~\ref{Sorites}.\ref{SoritesB}, so $i'<k$ and hence $(i',j')\not\in \llbracket p_k\rrbracket^{\mathcal{S}_{n,\delta}} $ by Fact~\ref{Atom}. Thus, $(i,j)\in \llbracket \neg p_k\rrbracket^{\mathcal{S}_{n,\delta}}$. Conversely, suppose $k<j$. Then $(k,\infty)\between (i,j)$ by Definition~\ref{Sorites}.\ref{SoritesB}, and $(k,\infty)\in \llbracket p_k\rrbracket^{\mathcal{S}_{n,\delta}}$ by Fact \ref{Atom}, so  $(i,j)\not\in \llbracket \neg p_k\rrbracket^{\mathcal{S}_{n,\delta}}$.
\qed\end{proof}

We can now show that in  $\mathcal{S}_{n,\delta}$, there are \textit{no sharp cutoffs} in the Sorites series.

\begin{fact}\label{NoSharpCutoffs} \textnormal{For any $k\in n$ and $\ell \in \delta+1$ with $k+\ell\leq n-1$, \[\mbox{$\llbracket p_k\wedge \neg p_{k+\ell}\rrbracket^{\mathcal{S}_{n,\delta}}= \varnothing$ and hence $\llbracket \neg (p_k\wedge \neg p_{k+\ell})\rrbracket^{\mathcal{S}_{n,\delta}}= S$.}\]}
\end{fact}

\begin{proof} For any $(i,j)\in S$,  we have $i+\delta<j$ by Definition~\ref{Sorites}.\ref{SoritesA}. Hence if ${(i,j)\in \llbracket p_k \rrbracket^{\mathcal{S}_{n,\delta}}}$, so $k\leq i$ by Definition~\ref{Sorites}.\ref{SoritesC}, then $k+\delta< j$ and hence $k+\ell <j$, so $(i,j)\not\in \llbracket \neg p_{k+\ell} \rrbracket^{\mathcal{S}_{n,\delta}}$ by Fact \ref{NegAtom}.
\qed\end{proof}

We can now put everything together to show the joint satisfiability of the claims that (i) the first element of the Sorites series satisfies the relevant predicate, (ii) the last element does not, and (iii) there are no sharp~cutoffs.

\begin{fact}\label{JointSat} \textnormal{For any $n,\delta\in\mathbb{N}$ with $1\leq \delta < n -1$, we have \[\mathcal{S}_{n,\delta}, (0,n-1)\Vdash p_0\wedge \neg p_{n-1}\wedge \underset{0\leq k\leq n- 2}{\bigwedge} \;\underset{\underset{k+\ell \leq n-1}{1\leq \ell \leq \delta}}{\bigwedge}\neg (p_k\wedge \neg p_{k+\ell}).\]}
\end{fact}
\begin{proof} Since $\delta <n -1$, we have $(0,n-1)\in S$ by Definition \ref{Sorites}.\ref{SoritesA}. Then we have $\mathcal{S}_{n,\delta}, (0,n-1)\Vdash p_0$ by Fact~\ref{Atom}; $\mathcal{S}_{n,\delta}, (0,n-1)\Vdash \neg p_{n-1}$ by Fact \ref{NegAtom}; and finally   $\mathcal{S}_{n,\delta}, (0,n-1)\Vdash \underset{0\leq k\leq n- 2}{\bigwedge} \;\underset{\underset{k+\ell \leq n-1}{1\leq \ell \leq\delta}}{\bigwedge}\neg (p_k\wedge \neg p_{k+\ell})$ by Fact~\ref{NoSharpCutoffs}.\qed\end{proof}

\begin{fact}\label{Consistency} \textnormal{For any $n,\delta\in\mathbb{N}$ with $1\leq \delta < n -1$, the formula \[p_0\wedge \neg p_{n-1}\wedge \underset{0\leq k\leq n- 2}{\bigwedge} \;\underset{\underset{k+\ell \leq n-1}{1\leq\ell\leq\delta}}{\bigwedge}\neg (p_k\wedge \neg p_{k+\ell})\] is consistent in orthologic and compatibility logic.}
\end{fact}
\begin{proof} Immediate from Fact \ref{JointSat}, Lemma \ref{EquivSem}, and Corollary \ref{SoundCor}.
\qed\end{proof}

The significance of Fact \ref{Consistency} is that we can consistently deny the existence of sharp cutoffs in a Sorites series by either (a) denying distributivity, while accepting excluded middle, as in orthologic or (b) denying excluded middle, while accepting distributivity, as in  compatibility logic. Which approach is preferable? If we endorse the arguments concerning epistemic modals mentioned in \S~\ref{Intro}, then accepting distributivity as a generally valid principle is not an available path. Still, we are left with the question of whether to accept excluded middle. 

Before turning to that question in \S~\ref{PseudosymmSorites}, let us address another natural question: is there something in common between the case of epistemic modals and the case of vague predicates that explains why an orthological treatment might be appropriate for both? The semantics used in this section provides a hint. Both epistemic modals and vagueness are related to \textit{partial states}, whether due to the partiality of information or the partiality of predicate determination. Moreover, these states arguably violate a key assumption underlying classical semantics identified in \cite{HM2022,Holliday-Mandelkern2022}: that \textit{compatibility implies compossibility}, where two states are compossible if they have a common refinement (recall Theorem~\ref{SpecialCases}.\ref{SpecialCases2}). In the case of epistemic modals, a state satisfying $\Diamond \neg p$ can be compatible with a state satisfying $p$, since $\Diamond \neg p$ does not entail $\neg p$, but they cannot have a common refinement, since such a refinement would satisfy $p\wedge\Diamond \neg p$, an epistemic contradiction (recall \S~\ref{LogicsSection}). In the case of vagueness, in particular in the model of the Sorites in Fig.~\ref{SymModel}, a state satisfying $p_1$ can be compatible with a state satisfying $\neg p_2$ (e.g., $(1,3)$ is compatible with $(0, 2)$), since $p_1$ does not entail $ p_2$, but they cannot have a common refinement, since such a refinement would satisfy $p_1\wedge\neg p_2$, contradicting our prohibition against sharp cutoffs in the Sorites. Thus, the cases of epistemic modals and vagueness are tied together by the relevant states violating the classical dictum that compatibility implies compossibility.

\section{A pseudosymmetric Sorites model}\label{PseudosymmSorites}

In this section, working with the fixpoint semantics of \S~\ref{FixpointSemantics}, we add to the symmetric Sorites model the possibility of \textit{rejecting} $p_k\vee\neg p_k$, for each $k\in n$. Following Definition \ref{AcceptRejectDef}, we say that a state $x$ rejects a formula $\varphi$ if for all $y\compflip x$, we have $\mathcal{M},y\nVdash \varphi$. In symmetric models, this is equivalent to $x$ forcing $\neg\varphi$, but in pseudosymmetric models, it is not. In particular, in pseudosymmetric models, it is possible for a state to reject an instance of excluded middle, as intuitionists may do. In the following construction, $k$ will be the state that rejects $p_k\vee\neg p_k$.\footnote{One could also add, for each interval $I\subseteq n$, a state $x_I$ that rejects $p_k\vee\neg p_k$ for each $k\in I$. But it suffices to make our points here to just add one state for each $p_k$.}

\begin{definition}\label{PseudoSorites} \textnormal{For $n,\delta\in\mathbb{N}$ with $1\leq\delta < n-1$, where $\mathcal{S}_{n,\delta}=(S,\between,V)$ is the symmetric Sorites model for $\mathcal{L}_n$ as in Definition~\ref{Sorites}, we define the \textit{pseudosymmetric Sorites model} $\mathfrak{S}_{n,\delta}=(X,\comp,V)$ as follows:
\begin{enumerate}
\item\label{PseudoSoritesA} $X=S\cup n$;
\item\label{PseudoSoritesB} if $(i,j),(i',j')\in S$, then $(i,j)\comp (i',j')$ iff $(i,j)\between (i',j')$;
\item\label{PseudoSoritesC} if $k\in n$ and $(i,j)\in S$, then $k\comp (i,j)$ iff $i<k<j$;
\item\label{PseudoSoritesD} if $x\in X$ and $k\in n$, then $x\comp k$.
\end{enumerate}
Note that the valuation $V$ in $\mathfrak{S}_{n,\delta}$ is the same as in $\mathcal{S}_{n,\delta}$.}
\end{definition}
The pseudosymmetric Sorites model for $n=4$ and $\delta=1$ is shown in Figure \ref{PseudosymModel}.

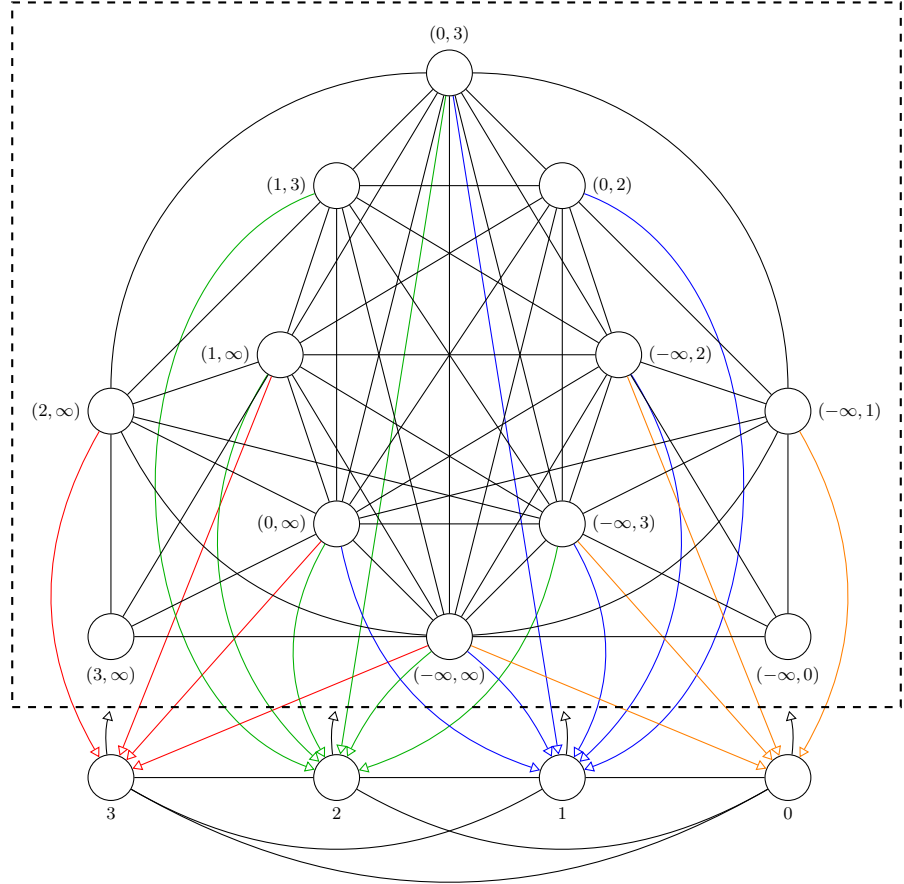
\begin{figure}[!h]
\begin{center}
\begin{tikzpicture}[scale=.75, every node/.style={scale=.75}][->,>=stealth',shorten >=1pt,auto,node distance=2.8cm,
                    semithick]

  \tikzstyle{every state}=[fill=white,draw=black,text=black]

  \node[state] (A) at (-6,0)                    [label=below:{$(3,\infty)$}] {{}};

  \node[state] (B) at (6,0)  [label=below:{$(-\infty, 0)$}] {{}};

  \node[state] (C) at (0,0) [label=below:{$(-\infty, \infty)$}] {{}};

  \node[state] (D) at (6,4)       [label=right:{$(-\infty, 1)$}] {{}};

  \node[state] (E) at (3,5)       [label=right:{$(-\infty, 2)$}] {{}};

  \node[state] (F) at (2,2)  [label=right:{$(-\infty, 3)$}] {{}};

  \node[state] (G) at (-2,2)  [label=left:{$(0, \infty)$}] {{}};

  \node[state] (H) at (2, 8)       [label=right:{$(0, 2)$}] {{}};

  \node[state] (I) at (0,10)  [label=above:{$(0, 3)$}] {{}};

  \node[state] (J) at (-3, 5)      [label=left:{$(1, \infty)$}] {{}};

  \node[state] (K) at (-2, 8)     [label=left:{$(1, 3)$}] {{}};

  \node[state] (L) at (-6, 4)      [label=left:{$(2, \infty)$}] {{}};
  
      \node[state] (M) at (-6, -2.5)      [label=below:{$3$}] {{}};
      \node[state] (N) at (-2, -2.5)      [label=below:{$2$}] {{}};
      \node[state] (O) at (2, -2.5)      [label=below:{$1$}] {{}};
      \node[state] (P) at (6, -2.5)      [label=below:{$0$}] {{}};
      
      \node  (AM) at (-6, -1.2)     {{}};
      \node  (AN) at (-2, -1.2)        {{}};
      \node (AO) at (2, -1.2)       {{}};
      \node  (AP) at (6, -1.2)       {{}};

\path 
	
	(C) edge[-{Triangle[open]},red] node {} (M)
	(G) edge[-{Triangle[open]},red] node {} (M)
	(J) edge[-{Triangle[open]},red] node {} (M)
	(L) edge[-{Triangle[open]},bend right,red] node {} (M)

	(C) edge[-{Triangle[open]},bend right=12,darkgreen] node {} (N)
	(J) edge[-{Triangle[open]},bend right=40,darkgreen] node {} (N)
	(F) edge[-{Triangle[open]},bend left,darkgreen] node {} (N)
	(G) edge[-{Triangle[open]},bend right,darkgreen] node {} (N)
	(K) edge[-{Triangle[open]},darkgreen,bend right=70] node {} (N)
	(I) edge[-{Triangle[open]},darkgreen] node {} (N)
	
	(C) edge[-{Triangle[open]},bend left=12,blue] node {} (O)
	(E) edge[-{Triangle[open]},bend left=40,blue] node {} (O)
	(F) edge[-{Triangle[open]},bend left,blue] node {} (O)
	(G) edge[-{Triangle[open]},bend right,blue] node {} (O)
	(H) edge[-{Triangle[open]},bend left=70,blue] node {} (O)
	(I) edge[-{Triangle[open]},blue] node {} (O)
	
	(C) edge[-{Triangle[open]},orange] node {} (P)
	(D) edge[-{Triangle[open]},bend left,orange] node {} (P)
	(E) edge[-{Triangle[open]},orange] node {} (P)
	(F) edge[-{Triangle[open]},orange] node {} (P)
	
	(M) edge[-{Triangle[open]},bend left=10] node {} (AM)
	(N) edge[-{Triangle[open]},bend left=10] node {} (AN)
	(O) edge[-{Triangle[open]},bend right=10] node {} (AO)
	(P) edge[-{Triangle[open]},bend right=10] node {} (AP)
	
	(M) edge[-] node {} (N)
	(M) edge[-,bend right] node {} (O)
	(M) edge[-,bend right] node {} (P)
	
	(N) edge[-] node {} (O)
	(N) edge[-,bend right] node {} (P)
	
	(O) edge[-] node {} (P)

    (A) edge[-]              node {} (C)
        edge[-]              node {} (G)
        edge[-]              node {} (J)
        edge[-]              node {} (L)
    (B) edge[-]              node {} (C)
        edge[-]              node {} (D)
        edge[-]              node {} (E)
        edge[-]              node {} (F)
    (C) edge[-, bend right]  node {} (D)
        edge[-]              node {} (E)
        edge[-]              node {} (F)
        edge[-]              node {} (G)
        edge[-]              node {} (H)
        edge[-]              node {} (I)
        edge[-]              node {} (J)
        edge[-]              node {} (K)
        edge[-, bend left]   node {} (L)
    (D) edge[-]              node {} (E)
        edge[-]              node {} (F)
        edge[-]              node {} (G)
        edge[-]              node {} (H)
        edge[-, bend right=45] node {} (I)
    (E) edge[-]              node {} (F)
        edge[-]              node {} (G)
        edge[-]              node {} (H)
        edge[-]              node {} (I)
        edge[-]              node {} (J)
        edge[-]              node {} (K)
    (F) edge[-]              node {} (G)
        edge[-]              node {} (H)
        edge[-]              node {} (I)
        edge[-]              node {} (J)
        edge[-]              node {} (K)
        edge[-]              node {} (L)
    (G) edge[-]              node {} (H)
        edge[-]              node {} (I)
        edge[-]              node {} (J)
        edge[-]              node {} (K)
        edge[-]              node {} (L)
    (H) edge[-]              node {} (I)
        edge[-]              node {} (J)
        edge[-]              node {} (K)
    (I) edge[-, bend right=45] node {} (L)
        edge[-]              node {} (J)
        edge[-]              node {} (K)
    (J) edge[-]              node {} (K)
        edge[-]              node {} (L)
    (K) edge[-]              node {} (L);
    
  \useasboundingbox (-7.75,-1.25) rectangle (8,11.25); 
  \draw [thick, dashed] (-7.75,-1.25) rectangle (8,11.25);
 
\end{tikzpicture}
\end{center}
\caption{The pseudosymmetric Sorites model for $n=4$ and $\delta=1$. We indicate $x\comp y$ by an arrow from $y$ to $x$. The arrow from $k$ to the dashed rectangle indicates arrows from $k$ to all states inside the rectangle. An edge with no arrow head indicates arrows in both direction.}\label{PseudosymModel}
\end{figure}

\begin{proposition}\label{Fixpoint2} \textnormal{$\mathfrak{S}_{n,\delta}$ is a fixpoint model.}
\end{proposition}

\begin{proof} We adapt the proof of Proposition \ref{Fixpoints} to show that $V(p_k)$ is a fixpoint of $c_\comp$. Suppose $x\not\in V(p_k)$.

Case 1: $x=(i,j)\in S$. Hence $i<j$ by Definition~\ref{Sorites}.\ref{SoritesA}. Since $(i,j)\not\in V(p_k)$, we have $i<k$ by Definition~\ref{Sorites}.\ref{SoritesC}. Let $(i',j')=(-\infty, k)$, so $(i',j')\in S$ by Definition~\ref{Sorites}.\ref{SoritesA}. Then $\mathrm{max}(i,i')=i<\mathrm{min}(j,j')$, so $(i',j')\between (i,j)$ by Definition~\ref{Sorites}.\ref{SoritesB} and hence $(i',j')\comp (i,j)$ by Definition~\ref{PseudoSorites}.\ref{PseudoSoritesB}. Now consider any $y\compflip (i',j')$. If $y=(i'',j'')\in S$, then $(i'',j'')\compflip (i',j')$ implies $(i'',j'')\between (i',j')$ by Definition~\ref{PseudoSorites}.\ref{PseudoSoritesB}, so we have $i''<j'=k$ by Definition \ref{Sorites}.\ref{SoritesB}, so $(i'',j'')\not\in V(p_k)$ by Definition \ref{Sorites}.\ref{SoritesC}. On the other hand, if $y=\ell \in n$, then $\ell \not\in V(p_k)$ since none of the new states from $n$ are in any $V(p_k)$. Thus, we have shown there is an $(i',j')\comp (i,j)$ such that for all $y\compflip (i',j')$, we have $y\not\in V(p_k)$.

Case 2: $x=\ell \in n$. Then again we take $(i',k')=(-\infty, k)$, so $(i',k')\comp \ell$ by Definition \ref{PseudoSorites}.\ref{PseudoSoritesD}. Then the rest of the argument is as in Case 1.

In either case, there is an $(i',j')\comp x$ such that for all $y\compflip (i',j')$, we have $y\not\in V(p_k)$. This shows that $V(p_k)$ is a fixpoint of $c_\comp$.\qed\end{proof}

Though the addition of the new states in $n$ breaks the symmetry of $\comp$ in $\mathfrak{S}_{n,\delta}$, we retain the property of pseudosymmetry needed for fundamental logic.

\begin{proposition} \textnormal{In $\mathfrak{S}_{n,\delta}=(X,\comp,V)$, $\comp$ is reflexive and pseudosymmetric, so fundamental logic is sound with respect to $\mathfrak{S}_{n,\delta}$ according to the semantics of Definition~\ref{TruthDef}.}
\end{proposition}
\begin{proof} Reflexivity is obvious. For pseudosymmetry, by Definition \ref{PseudoSorites}.\ref{PseudoSoritesB}, the restriction of $\comp$ to $S$ is symmetric; by Definition \ref{PseudoSorites}.\ref{PseudoSoritesD}, the restriction of $\comp$ to $n$ is symmetric; and by Definition \ref{PseudoSorites}.\ref{PseudoSoritesD}, $k \comp (i,j)$ implies $(i,j)\comp k$. Thus, we need only consider the case where $(i,j)\comp k$. In this case, we simply observe that $(i,j)\compflip (i,j)$ by Definitions \ref{PseudoSorites}.\ref{PseudoSoritesB} and \ref{Sorites}.\ref{SoritesB}, and $(i,j)$ pre-refines $k$ by Definition~\ref{PseudoSorites}.\ref{PseudoSoritesD}. This establishes pseudosymmetry.\qed\end{proof}

We now prove a series of facts analogous to those for $\mathcal{S}_{n,\delta}$  in \S~\ref{SymmSorites}.

\begin{fact}\label{Atom2} \textnormal{For any $k\in n$, $\llbracket  p_k\rrbracket^{\mathfrak{S}_{n,\delta}}  = \{ ( i,j ) \in S\mid k\leq i\}$.} 
\end{fact}
\begin{proof} Immediate from  Definitions \ref{TruthDef}.\ref{TruthDefAtom}, \ref{PseudoSorites}, and \ref{Sorites}.\ref{SoritesC}.
\end{proof}

\begin{fact}\label{NegAtom2} \textnormal{For any $k\in n$, \[\mbox{$\llbracket \neg p_k\rrbracket^{\mathfrak{S}_{n,\delta}} = \{(i,j)\in S\mid  j\leq k\}$ and $\llbracket \neg\neg p_k\rrbracket^{\mathfrak{S}_{n,\delta}} =\llbracket p_k\rrbracket^{\mathfrak{S}_{n,\delta}}$.}\]}
\end{fact}
\begin{proof} The proof of the first equation is essentially the same as that of Fact \ref{NegAtom}, only using the facts that (i) none of the new states $\ell\in n$ belong to any $V(p_k)$, and (ii) the new states $\ell\in n$ are such that $x\comp \ell$ for all $x\in X$.

 For the second equation, we have $\llbracket  p_k\rrbracket^{\mathfrak{S}_{n,\delta}}  = \{ ( i,j ) \in S\mid k\leq i\}$ by Fact~\ref{Atom2}. Now if $(i,j)\in S$ and $k\leq i$, then for any $(i',j')\comp (i,j)$, we have $k<j'$ by Definitions \ref{PseudoSorites}.\ref{PseudoSoritesB} and \ref{Sorites}.\ref{SoritesB}, so $(i',j')\not\in \llbracket \neg p_k\rrbracket^{\mathfrak{S}_{n,\delta}} $, which shows that ${(i,j)\in \llbracket \neg\neg p_k\rrbracket^{\mathfrak{S}_{n,\delta}}}$. Conversely, if $i<k$, then $(-\infty,k)\comp (i,j)$ by Definitions \ref{PseudoSorites}.\ref{PseudoSoritesB} and \ref{Sorites}.\ref{SoritesB}, and $(-\infty,k)\in \llbracket \neg p_k\rrbracket^{\mathfrak{S}_{n,\delta}}$, so $(i,j)\not\in \llbracket \neg\neg p_k\rrbracket^{\mathfrak{S}_{n,\delta}}$. Similarly, for each $k\in n$, we have $(-\infty, k)\comp$ by Definition \ref{PseudoSorites}.\ref{PseudoSoritesD}, so $k\not\in \llbracket \neg\neg p_k\rrbracket^{\mathfrak{S}_{n,\delta}}$.\qed\end{proof}

\begin{fact}\label{NoSharpCutoffs2} \textnormal{For any $k\in n$ and $\ell \in \delta+1$ with $k+\ell\leq n-1$, \[\mbox{$\llbracket p_k\wedge \neg p_{k+\ell}\rrbracket^{\mathfrak{S}_{n,\delta}}= \varnothing$ and hence $\llbracket \neg (p_k\wedge \neg p_{k+\ell})\rrbracket^{\mathfrak{S}_{n,\delta}}= X$.}\]}
\end{fact}
\begin{proof} The proof is the same as that of Fact \ref{NoSharpCutoffs}, only using the fact that none of the new states $k\in n$ belong to any $V(p_k)$.
\qed\end{proof}
\noindent Thus, the analogue of Fact \ref{JointSat} for $\mathfrak{S}_{n,\delta}$ also holds.

What is new about $\mathfrak{S}_{n,\delta}$ is that excluded middle can now be used to define the set of states at which the status of $p_k$ is settled. 

\begin{fact}\label{EMExt} \textnormal{For any $k\in n$, we have $\llbracket p_k\vee\neg p_k\rrbracket^{\mathfrak{S}_{n,\delta}}=\{(i,j)\in S\mid k\leq i\mbox{ or }j\leq k\}$.}
\end{fact}

\begin{proof} The right-to-left inclusion is immediate from Facts \ref{Atom2} and \ref{NegAtom2}. For the left-to-right inclusion, suppose $x\in X\setminus \{(i,j)\in S\mid k\leq i\mbox{ or }j\leq k\}$. Then either $x\in n$ or $x=(i,j)\in S$ and  $i<k<j$. In either case, we have $k\comp x$ by Definitions \ref{PseudoSorites}.\ref{PseudoSoritesC} and \ref{PseudoSorites}.\ref{PseudoSoritesD}.  Moreover, for any $y\compflip k$, either $y=k$ by Definition~\ref{PseudoSorites}.\ref{PseudoSoritesD} or $y=(i',j')$ where $i'<k<j'$ by Definition~\ref{PseudoSorites}.\ref{PseudoSoritesC}, and in either case, $y\not\in \llbracket  p_k\rrbracket^{\mathfrak{S}_{n,\delta}} \cup \llbracket  \neg p_k\rrbracket^{\mathfrak{S}_{n,\delta}} $ by Facts \ref{Atom2} and \ref{NegAtom2}. It follows by Definition \ref{TruthDef}.\ref{TruthDefOr} that $x\not\in \llbracket p_k\vee\neg p_k\rrbracket^{\mathfrak{S}_{n,\delta}}$.\qed\end{proof}

In addition, $\mathfrak{S}_{n,\delta}$ exhibits one of Fine's \cite[p.~44ff]{Fine2020} main claims about the ``global'' nature of indeterminacy in a Sorites series: although negating a single instance of excluded middle is of course contradictory, when dealing with a Sorites series we should be able to negate a conjunction of distinct instances of excluded middle or even weak excluded middle (i.e., $\neg\varphi\vee\neg\neg\varphi$).\footnote{Fine \cite[p.~34]{Fine2020} suggests that ``under fairly innocuous logical assumptions'', the negation of the conjunction of instances of excluded middle is equivalent to the conjunction of negations of state descriptions of the form $p_1\wedge\dots\wedge p_k\wedge \neg p_{k+1}\wedge\dots\wedge\neg p_n$. But orthologic makes the former unsatisfiable while the latter satisfiable (Fact \ref{JointSat}), and arguably distributivity is not an innocuous logical assumption.}

\begin{fact}\label{DenyLEMs} \textnormal{In $\mathfrak{S}_{n,\delta}$, we have
\begin{eqnarray*}
&&\llbracket(p_0\vee\neg p_0)\wedge\dots\wedge (p_{n-1}\vee\neg p_{n-1})\rrbracket^{\mathfrak{S}_{n,\delta}}\\&=&\llbracket(\neg p_0\vee\neg \neg p_0)\wedge\dots\wedge (\neg p_{n-1}\vee\neg \neg p_{n-1})\rrbracket^{\mathfrak{S}_{n,\delta}}\\
&=&\{(n-1,\infty), (-\infty,n-1)\},
\end{eqnarray*} so 
\begin{eqnarray*}
&&\llbracket \neg ((p_0\vee\neg p_0)\wedge\dots\wedge (p_{n-1}\vee\neg p_{n-1}))\rrbracket^{\mathfrak{S}_{n,\delta}}\\
&=&\llbracket \neg ((\neg p_0\vee\neg\neg p_0)\wedge\dots\wedge (\neg p_{n-1}\vee\neg\neg p_{n-1}))\rrbracket^{\mathfrak{S}_{n,\delta}}\\
&=&\{(i,j)\in S \mid i,j\in\mathbb{N} \}.
\end{eqnarray*}}
\end{fact}
\begin{proof} We already know the formulas using excluded middle and weak excluded middle have the same extension by the second part of Fact \ref{NegAtom2}. 

Now if $x\in \llbracket(p_0\vee\neg p_0)\wedge\dots\wedge (p_{n-1}\vee\neg p_{n-1})\rrbracket^{\mathfrak{S}_{n,\delta}}$, then by Fact \ref{EMExt}, $x=(i,j)\in S$ and for every $k\in n$, $k\leq i$ or $j\leq k$. It follows that $i=n-1$ or $j=n-1$, which implies $x=(n-1,\infty)$ or $x=(-\infty,n-1)$ by Definition \ref{Sorites}.\ref{SoritesA}. For the second part, the only states $y$ for which neither $(n-1,\infty)\comp y$ nor $(-\infty,n-1)\comp y$ are the pairs of natural numbers. Hence $\llbracket \neg ((p_0\vee\neg p_0)\wedge\dots\wedge (p_{n-1}\vee\neg p_{n-1}))\rrbracket^{\mathfrak{S}_{n,\delta}}=\{(i,j)\in S \mid i,j\in\mathbb{N} \}$.\qed\end{proof}

\section{Discussion}\label{Discussion}

Let us now return to our initial question: what is a natural non-classical base logic to which to retreat in light of both the non-classicality emerging from epistemic modals and the non-classicality emerging from vagueness? We are simply taking for granted here the arguments against distributivity involving epistemic modals from \cite{Holliday-Mandelkern2022}. As we have seen in \S~\ref{SymmSorites}, dropping distributivity as in orthologic is also enough to render consistent the denial of sharp cutoffs in the Sorites. Thus, a temptingly economical view is that the key to dealing with both the non-classicality emerging from epistemic modals and the non-classicality emerging from vagueness is to deny distributivity and retreat to orthologic.

However, we have also seen in \S~\ref{PseudosymmSorites} a potential benefit of going still weaker than orthologic, down to fundamental logic. In the fundamental approach to the Sorites, as in Fine's approach, instances of excluded middle such as $p_k\vee\neg p_k$ have genuine expressive value (recall Fact \ref{EMExt}): rather than being empty tautologies, as in orthologic, in the fundamental approach they express that \textit{there is a fact of the matter} about whether a given person is young or not, whether they are bald or not, etc., a function for excluded middle highlighted by Field \cite{Field2003}. Thus, we have the ability to express for some questions---e.g., whether there is phosphene on Venus---that there is a fact of the matter ($phosphene\vee\neg phosphene$), while withholding such claims for other questions---e.g., whether Jerry Fallwell's life began on a particular nanosecond, to borrow Field's example. It is true that we still cannot assert $\neg (p\vee\neg p)$, which is inconsistent, but we are able to withhold assent from $p\vee\neg p$; and as Fine \cite[p.~39f]{Fine2020} stresses, we can even use negation to deny that for every item in a Sorites series, there is a fact of the matter about whether the relevant predicate applies (recall Fact~\ref{DenyLEMs}). 

From Fine's point of view, there would seem to be another objection to orthologic. In orthologic, $\neg (p_k\wedge\neg p_{k+1})$ is equivalent to $\neg p_k\vee p_{k+1}$, so orthologic must treat the conjunctive and disjunctive versions of the Sorites as equivalent. Yet Fine \cite[p.~737]{Fine2023} wants to treat them differently: for the disjunctive version, disjunctive syllogism is valid on Fine's approach, but he does not accept the Sorites premises of the form $\neg p_k\vee p_{k+1}$. By contrast, in orthologic, $\neg p_k\vee p_{k+1}$ is equivalent to $\neg (p_k\wedge\neg p_{k+1})$, which Fine accepts, but disjunctive syllogism is not valid in orthologic.\footnote{For example, in the epistemic orthologic of \cite{Holliday-Mandelkern2022}, one can accept $p\vee \Box \neg p$ (e.g., ``Either the cat is inside or he must be outside'') and $\Diamond p$ (``The cat might be inside''), without concluding $p$ (``The cat is inside'') (cf.~\cite[citing Yalcin]{KlinedinstRothschild2012}).} As in Fine's approach, in the fundamental approach of \S~\ref{PseudosymmSorites}, $\neg p_k\vee p_{k+1}$ is not equivalent to $\neg (p_k\wedge \neg p_{k+1})$; while the latter is forced at all states in $\mathfrak{S}_{n,\delta}$, $\neg p_0\vee p_1$ is not forced at, e.g., the state $(0,3)$.\footnote{This is because $1\comp (0,3)$, and for any $x\compflip 1$, $x\not\in \llbracket p_1\rrbracket^{\mathfrak{S}_{n,\delta}}$ and $x\not\in \llbracket \neg p_1\rrbracket^{\mathfrak{S}_{n,\delta}}$. Then since $x\in \llbracket \neg p_0\rrbracket^{\mathfrak{S}_{n,\delta}}$ implies $x\in \llbracket \neg p_1\rrbracket^{\mathfrak{S}_{n,\delta}}$ in $\mathfrak{S}_{n,\delta}$, we also have $x\not\in \llbracket \neg p_0\rrbracket^{\mathfrak{S}_{n,\delta}}$.}  Indeed, in $\mathfrak{S}_{n,\delta}$, a state forces $\neg p_k\vee p_{k+1}$ only if it forces $\neg p_k$ or forces $p_{k+1}$. Thus, although disjunctive syllogism is not schematically valid in fundamental logic, in the model $\mathfrak{S}_{n,\delta}$, any state that forces $p_k\wedge (\neg p_k\vee p_{k+1})$ also forces $p_{k+1}$.

Although we did not include a conditional in our language in this paper, if we were to add a conditional, this might provide another argument in favor of the fundamental approach over the orthological one. Suppose we accept that the ``Or-to-If'' inference from $\neg p_k\vee p_{k+1}$ to $p_k\to p_{k+1}$ preserves certainty.\footnote{Arguably the inference is not \textit{valid}, assuming the standard view that probability is monotonic with respect to logical consequence, as shown by the fact that one can rationally assign higher probability to $\neg p\vee q$ than to $p \to q$. To use an example from Bas van Fraassen, it is rational to assign `The die did not land on an odd number or it landed on 5' probability 4/6 and `If the die landed on an odd number, then it landed on 5' probability 1/3.} Then if we are certain there is no sharp cutoff, the orthological approach would commit us to certainty in $\neg p_k\vee p_{k+1}$ for each $k$, in which case Or-to-If would commit us to certainty in $p_k\to p_{k+1}$ for each $k$; hence we would have to deny modus ponens for $\to$ in order to block the derivation of $p_n$. Though there may be reasons to deny modus ponens for conditionals whose consequents contain epistemic modals or conditionals \cite{McGee1985}, being forced to deny modus ponens for simple conditionals---albeit with vague predicates---may seem more costly. Perhaps a proponent of the orthological approach to vagueness could escape this modus ponens problem by denying that we are certain that there are no sharp cutoffs, or by denying that Or-to-If preserves certainty when vague predicates are involved, inspired by Fine's rejection of the move  from $\neg (p_k\wedge\neg p_{k+1})$ to $p_k\to p_{k+1}$. Whichever path one chooses, the cost-benefit analysis seems likely to be quite subtle.

Let us suppose for the moment that one is convinced by some of the considerations above to give up the orthological principles that are rejected by the intuitionists, namely excluded middle and the inference from $\neg(\varphi\wedge \neg\psi)$ to $\neg\varphi\vee\psi$. Combined with the arguments against distributivity (or proof-by-cases with side assumptions) involving epistemic modals mentioned in \S~\ref{Intro}, this would lead us toward fundamental logic. Still, one could admit that certain principles that are not generally valid are safe when restricted to propositions of special types. For example, one could hold that formulas without epistemic modals express propositions in a special subalgebra of the ambient algebra of propositions; in \cite{Holliday-Mandelkern2022} this is a Boolean algebra, but if vague predicates are allowed, it could instead be a bounded distributive lattice with a weak pseudocomplementation, $(L,\neg)$. Then formulas containing neither epistemic modals nor vague predicates could be taken to express propositions  in a Boolean subalgebra of $(L,\neg)$.

In a sense, nothing is lost by moving from orthologic to fundamental logic. For the G\"{o}del-Gentzen translation from classical logic to intuitionistic logic \cite{Godel1933,Gentzen1936} is also a full and faithful embedding of orthologic into fundamental logic. Recall this is the translation $g$ defined by 
\begin{eqnarray*}
g(p)&=&\neg\neg p \\
g(\neg\varphi)&=&\neg g(\varphi)\\
g(\varphi\wedge\psi)&=&(g(\varphi)\wedge g(\psi))\\
g(\varphi\vee\psi)&=&g(\neg(\neg\varphi\wedge\neg\psi)).
\end{eqnarray*}
 
 \begin{proposition}[\cite{Holliday2023}] \textnormal{For all $\varphi,\psi\in\mathcal{L}$, we have $\varphi\vdash\psi$ in orthologic if and only if $g(\varphi)\vdash g(\psi)$ in fundamental logic.}
 \end{proposition}
 \noindent Thus, all orthological reasoning can be carried out inside fundamental logic,\footnote{In fact, all classical propositional reasoning can also be carried out inside fundamental logic at the expense of an exponential blowup in the ``translation''; see Proposition~2.4 of \cite{Holliday2024}.} while fundamental logic has the advantage that orthological tautologies such as excluded middle have genuine expressive value---and hence withholding assent from them also has genuine expressive value.
 
Might some other sublogic of orthologic and compatibility logic be appropriate for handling the non-classicality coming from epistemic modals and vagueness? Consider, for example, the \textit{intersection} of orthologic and compatibility logic: $\varphi\vdash\psi$ in the intersection logic iff $\varphi\vdash\psi$ in both orthologic and compatibility logic. This logic is stronger than fundamental logic, as it includes principles such as $p\wedge (q\vee r)\vdash (p\vee\neg p)\vee ((p\wedge q)\vee (p\wedge r))$, since the left disjunct of the conclusion follows in orthologic and the right disjunct follows in compatibility logic. However, it is doubtful that there is a well-motivated proof theory or natural semantics for this intersection logic. By contrast, fundamental logic has a well-motivated Fitch-style proof theory based on introduction and elimination rules \cite{Holliday2023}, as well as a sequent calculus \cite{Aguilera2022}, and a natural semantics using reflexive and pseudosymmetric frames. Until a similarly natural sublogic of orthologic and compatibility logic appears, it appears that fundamental logic is the most natural sublogic of orthologic and compatibility logic to consider for accommodating the non-classicality coming from epistemic modals and from vagueness.

\section{Conclusion}\label{Conclusion}

It is a tantalizing possibility that the system of orthologic first discovered in the context of quantum mechanics \cite{Birkhoff1936} could solve both the puzzle of epistemic modals \cite{Holliday-Mandelkern2022} and the paradox of the Sorites. Whether vagueness also provides motivation to go weaker than orthologic, perhaps down to the system of fundamental logic, is a question we have discussed but not settled. Of course, there are other motivations for such a weakening, coming from the tradition of constructive mathematics \cite{Troelstra1988a}. In fact, some reasons for rejecting excluded middle in mathematical contexts are not entirely unrelated to reasons having to do with vagueness. For example, Feferman \cite{Feferman2011} rejects the instance of excluded middle for the Continuum Hypothesis on the grounds that ``the concept of the totality of arbitrary subsets of A is essentially underdetermined or vague'' (p.~21).

Here we have drawn inspiration from Fine's \cite{Fine2020} non-classical approach to vagueness. There are of course other approaches to vagueness that attempt to salvage more of classical logic, such as supervaluationist and epistemicist approaches (see \cite{Williamson1994}), or at least to hold the line at intuitionistic logic \cite{Wright2001,Bobzien2020}. The orthological and fundamental approaches to vagueness should be systematically compared to these other approaches. But those who accept (\ref{SoritesForm}) concerning the Sorites Paradox cannot accept the classical or even intuitionistic views of $\wedge$ and~$\neg$. We hope the models discussed in this paper render more intelligible the consistency of (\ref{SoritesForm}) according to alternative views of the connectives.

\subsection*{Acknowledgements}
I thank the two reviewers for valuable comments, as well as Ahmee Christensen, Matt Mandelkern, and Guillaume Massas for helpful discussion. I also thank the organizers of the 4th Tsinghua Interdisciplinary Workshop on Logic, Language, and Meaning for the invitation to speak on the topic of The Connectives in Logic and Language, which prompted this paper. 

\bibliographystyle{splncs04}
\bibliography{fundamental}

\end{document}